%% file: main.tex
\pdfoutput=1 
\documentclass[nonacm]{acmart}

\AtBeginDocument{%
  }

\setcopyright{acmlicensed}
\copyrightyear{2025}
\acmYear{2025}
\acmDOI{XXXXXXX.XXXXXXX}

\acmConference[WWW '26]{}{April 13--17,
  2026}{Dubai, UAE}

\makeatletter
\def\ps@pprintTitle{%
  \let\@oddhead\@empty
  \let\@evenhead\@empty
  \let\@oddfoot\@empty
  \let\@evenfoot\@oddfoot
}
\makeatother
\input{preamble}

\begin{document}
\title{A Small  Collusion is All You Need}

\author{Yotam Gafni}
\affiliation{%
 \institution{Weizmann Institute of Science}
 \city{Rehovot 7630031}
 \country{Israel}}


\begin{abstract}
Transaction Fee Mechanisms (TFMs) study auction design in the Blockchain context, and emphasize robustness against miner and user collusion, moreso than traditional auction theory. \cite{chung2023foundations} introduce the notion of a mechanism being $c$-Side-Contract-Proof ($c$-SCP), i.e., robust to a collusion of the miner and $c$ users. Later work \cite{chung2024collusion,welfareIncreasingCollusion} shows a gap between the $1$-SCP and $2$-SCP classes. We show that the class of $2$-SCP mechanisms equals that of any $c$-SCP with $c\geq 2$, under a relatively minor assumption of consistent tie-breaking. In essence, this implies that any mechanism vulnerable to collusion, is also vulnerable to a small collusion. 
\end{abstract}


\maketitle

\section{Introduction}

Transaction Fee Mechanisms (TFMs) is the name given to the class of mechanisms responsible for the ordinary operation of a blockchain system, choosing which transactions are confirmed at each issued block. In Bitcoin, this is done by having each transaction commit to a fee that it pays if confirmed. This is in essence a pay-as-bid auction, and a line of work since \cite{lavi2019redesigning,huberman2021monopoly} has considered alternative auction designs. Importantly, this has not remained a theoretical exercise, and major blockchain protocols have adopted different mechanisms, most notably Ethereum, which adopted the EIP-1559 mechanism (in essence, a first-price with a reserve that is fully burned) \cite{roughgarden2020eip1559}. 

TFMs have several unique characteristics that distinguish them from traditional auctions \cite{roughgarden8reasons}: (i) They allow for payments to be ``burned'', i.e., the user pays but the miner (which takes the role of the auctioneer) does not receive the full payment, (ii) Their research is more focused on \textit{robustness} to different forms of attacks, rather than just incentive-compatibility w.r.t. the bidders, and welfare/revenue maximization within this framework, which is the standard classic approach to auctions \cite{vickrey1961, myerson1981optimal}. In particular, the issue of \textit{collusion-resistance} has taken an important role in TFM research \cite{roughgarden2020transaction,chung2023foundations}. Even in traditional auctions, practitioners have long emphasized the importance of collusion-resistance. The following is an emblematic quote:

\begin{displayquote}
``The most important features of an auction are its robustness against collusion and its attractiveness to potential bidders. Failure to attend to these issues can lead to disaster.'' \cite{klempererWhatMatters}
\end{displayquote}


However, this takes maybe even greater precedence in Blockchain systems, because of their online and ephemeral nature. In an anonymous, unregulated environment, bidders may worry less about facing prosecution for cartel behavior, as they have faced many times in traditional auctions \cite{posner70}. Thus, TFM research strives to make collusion-resistance a built-in part of its design.

We now arrive at the topic at the center of this paper. 
\cite{gafni2022greedy} raise the question, supported by influential work in economic theory \cite{Stigler1964}, of whether collusion is a real risk in TFMs. The argument goes, that since in some systems block times are quick (e.g., $12$ seconds in Ethereum), it is hard to not only establish a collusion, but also coordinate, through possibly several rounds of communication, how each agent should partake in it. Both in traditional auctions and in TFMs, there is some support in theory and empirics that the \textit{smaller} the collusion, the \textit{easier} it is to establish it, coordinate it, and avoid its detection. Thus, it is important to know whether mechanisms are robust to \textit{small}-size collusion. 

Our work does exactly that: We consider the popular collusion notion of $c$-SCP, introduced by \cite{chung2023foundations}. The notion considers whether a coalition of the miner and up to $c$ bidders is able to increase its utility by changing bids, as well as using manipulations by the miner (dropping bids, or adding shill bids). The $c$ parameter naturally establishes a hierarchy: Clearly, the class of $1$-SCP mechanisms (i.e., these robust to a collusion of the miner and at most a single user) is the widest class, while the class of SCP (used as shorthand for $n$-SCP, i.e., a collusion of the miner and any number of bidders) is the narrowest. Moreover, 
an example of \cite{chung2024collusion} shows a gap between the class of $1$-SCP and $2$-SCP mechanisms. It is thus natural to ask, 

\begin{displayquote}
\textbf{Question:} 
Is there a separation between each $c$-SCP and $(c+1)$-SCP class? 
\end{displayquote}

We answer this question in the negative, and show that other than the $1$-SCP and $2$-SCP gap, there are no more gaps, and in terms of mechanism classes, $2$-SCP $=$ SCP. Moreover, if we can find an arbitrary collusion against a mechanism, then it is reducible in poly-time to a simple one, using our explicit reduction. Thus, at least in terms of the need to involve many parties in the collusion, we show that this is not an issue. We detail remaining hurdles for the colluders in the discussion.

\subsection{Related Work}

\subsubsection{The Core TFM Framework and Extensions}

Considering blockchain transaction fee mechanisms as an auction design problem was studied by \cite{lavi2019redesigning,huberman2021monopoly}. The early emphasis was on \textit{miner} incentive-compatibility (MIC), i.e., robustness to manipulations by the miner such as dropping bids, and introducing fake bids. \cite{roughgarden2020transaction} was the first to introduce collusion-resilience as the centerpiece of TFM design, through the notion of \textit{OCA}-proof. OCA-proof is a notion that requires the mechanism to be welfare-maximizing for the \textit{grand coalition}
. \cite{chung2023foundations} introduce the $c$-SCP notion, which captures a more game-theoretic group deviation notion, which we study in this paper. The focus of the axiomatic literature that followed was on showing whether a ``dream'' TFM, that combines collusion-resilience, as well as the user and miner incentive-compatibility notions, exists. The overall conclusion is that the answer is no (To mention some results, no user incentive-compatible (UIC) + $1$-SCP mechanism \cite{chung2023foundations} exists, no UIC+MIC+\textit{OCA}-proof deterministic mechanism, scalable randomized one, or randomized where the collusion knows the random coins \cite{chung2024collusion,gafni2024barriers} exists), although some open questions remain. In particular, a non-scalable randomized mechanism where the collusion is oblivious to the random coins may still exist, though its welfare is bounded away from optimal \cite{gafni2024barriers}. However, one might be interested in a subset of the restrictions. For example, if we think that regular users are ephemeral and have bounded rationality, we may omit the user incentive-compatibility requirement. One reason to think so is the common use of fee estimators and calculators by users \cite{zhang2024transactionfeeestimationbitcoin}). With only a subset of restrictions, there are characterizations of the classes of mechanisms that implement them \cite{gafni2024barriers}. Importantly, we are not aware of a past work that examined the notion of SCP in isolation, and drew conclusions and characterizations from it alone, and this is a novelty of our work.

Beyond the core framework, the study of TFMs was extended in many directions: (1) Considering time preferences \cite{gafni2024scheduling,nisan2023serial,babaioff2024optimality}, (2) Considering the problem of Maximal Extractable Value (MEV), i.e., incentive problems that may arise in certain blockchain use-cases where the order of the transactions in the block matters \cite{bahrani2024transaction}, 
(3) Considering more advanced blockchain architectures, such as DAGs and PBS \cite{garimidiPBS,garimidi2025dag}. We emphasize that our work is in the core TFM setting, and does not consider these variations. 


In the appendix, we overview additional related work that details the theory and empirical evidence of collusion, both in Blockchains and traditional auctions.

\subsection{An Overview of Our Proof}

Our proof is constructive and presents a reduction from a given general collusion to a $2$-collusion. First, we make sure the collusion is in what we call ``Canonical Ordering''. This means that bidders in the collusion do not change their relative order, in terms of value, before and after the collusion. This is possible since a collusion can always have even better utility by allowing higher-value agents to be confirmed. Keeping the canonical ordering throughout our reduction is useful in making the problem one-dimensional ('how many colluding bidders are confirmed?') instead of multi-dimensional ('is each of the colluding bidders confirmed?'). 
The main idea is to decompose the collusion into steps, where at each step only one bidder changes their bid. The goal is to isolate one step which is in itself a beneficial collusion. However, this has to be done carefully
. Two important insights allow the decomposition:

(i) To maintain $2$-SCP, there must be a monotonicity of the number of confirmed bidders in whether a bidder has increased / decreased their bid. If a bidder in the collusion increases their bid, the number of confirmed bidders must weakly grow, and vice versa. In some sense, this is like incentive-compatibility for a single bidder, but applied to the set of all the colluders.  There are however some subtleties, since the miner is also part of the joint utility of the collusion, and also we need to only use the $2$-SCP condition (and not the more general SCP). 

(ii) We aim to have some kind of a telescopic argument: If you have a decomposition of the collusion into one-by-one steps, then one of the steps is a collusion in itself. For this to work, the changes at each step should accumulate in a way that corresponds to the shift from the original $A$ to $B$. However, the changes may be incorrectly accounted for in this shift: Every time a bidder changes their bid and is confirmed, the difference in values is accounted into the sum, but these may be different than the bidders confirmed in $B$, whose difference of values is accounted for in the collusion. This goes both ways: There may be bidders confirmed in the decomposition that are not confirmed in $B$, and vice versa. 
Thus, we need to make sure, as Lemma~\ref{lem:car_parking} requires, that ``increasing and ultimately unconfirmed'' bidders are not confirmed throughout the decomposition, and ``decreasing and ultimately confirmed'' are confirmed in the decomposition. 
It turns out that this is guaranteed with a certain decomposition, where the  ``increasing and confirmed'' bidders move first, then all decreasing bidders move, and finally the ``increasing but unconfirmed'' bidders move. We nick-name this construction ``The Salsa Dance Decomposition'' because of the resemblance to the forward-backward-forward nature of the basic Salsa dance step. 

We are thus able to reduce any collusion to a collusion where a single bidder moves, and all other parties to the collusion are passive beneficiaries of this change. It then remains to isolate a single beneficiary of the change, and have the mover and the beneficiary would constitute the collusion, together with the miner. This in itself is not trivial since the $2$-collusion needs improve its joint utility: For example, you could have a ``mover'' that loses utility of $10$ due to the move, and $3$ beneficiaries that gain $4$ each. While it is easy to isolate a single beneficiary (take any of them), this will not suffice for a $2$-collusion. This again requires making small steps, this time w.r.t. how the mover changes their bid. We show that by gradually increasing the bid we can isolate a change where the mover has a small utility loss, but at least one of the beneficiaries has a larger utility gain.

\section{Model and Preliminaries}
\label{sec:Model}
We follow the standard model used by the literature for Transaction Fee Mechanisms (TFMs) \citep{lavi2019redesigning,roughgarden2024transaction,chung2023foundations,shi2023what}. The model assumes a multi-unit, unit-demand setting, i.e., bidders have a fixed-size transaction, and do not care about specific placement in the block. 
We interchangeably use the terms \emph{auction}, \emph{mechanism}, or \gls{TFM} to refer to the object of our discussion, depending on what is most appropriate in the context. 


\subsection{Transaction Fee Mechanisms}
A \gls{TFM} consists of a \emph{confirmation rule} $\alloc$, a \emph{payment rule} $\pay$ and a \emph{burning rule} $\burn$.

\paragraph*{Confirmation rule}
A confirmation rule defines the mechanism's confirmation of transactions to the upcoming block, as intended by the mechanism's designers.
\begin{definition}[Confirmation rule]
    \label{def:AllocFunc}
    A \emph{deterministic} confirmation rule $\alloc:R_+^n \rightarrow \{0,1\}^n$ defines which transaction should be included in the upcoming block.
    Generally, a confirmation rule may be \emph{randomized}, i.e., $\alloc:R_+^n \rightarrow \{\Delta(0,1)\}^n$. 
\end{definition}

\paragraph*{Payment \& burn}
The payment and burn rules respectively define the amount of fees paid by transactions when included in a block, and how much of each payment is ``burnt'' and taken out of circulation instead of given to the miner.

\begin{definition}[Payment rule]
    The payment rule $\mathbf{\pay}:R_+^n\rightarrow R_+^n$ receives bids $b_1, \ldots, b_n$, and outputs a payment vector with a payment for each bidder $i$.
\end{definition}
\begin{definition}[Burning rule]
    The burning rule $\mathbf{\burn}:R_+^n\rightarrow R_+^n$ in auction $\auction$ receives bids $b_1, \ldots, b_n$, and and outputs a burn vector with a burn for each bidder $i$, meaning the amount of funds out of $b_i$'s payment that is taken out of circulation
    . 
\end{definition}

\begin{definition}[Agent Utilities]
    \label{def:utilities}
    Let auction $\auction = (\alloc, \pay, \burn)$. Consider $n$ bidders with true valuations $\mathbf{v} \in R_+^n$ and $n'$ bids 
    $\mathbf{b} \in R_+^{n'}$. 
    W.l.o.g., $\mathbf{b}$ is indexed the same as $\mathbf{v}$, meaning that the first $n$ bids in $\mathbf{b}$ correspond to the bids in $\mathbf{v}$ (possibly as $0$ if omitted), thus w.l.o.g. let $n' \geq n$.
    We consider the $n' - n$ bidders that are in $\mathbf{b}$ but not in $\mathbf{v}$ as the miner's ``fake bids'', and denote them by $\mathbf{b}_F$. 
    Then, the utilities of the various agents are defined as follows.
    \begin{itemize}
    \item  \emph{Bidder Utility:}
    \begin{equation}
    \label{eq:bidder_util}
        u_i(\fee_i, \mathbf{\fee_{-i}} ; v_i)
        \define
            v_i \cdot \alloc(b_i, \mathbf{b}_{-i}) - \pay\left(\fee_i,\mathbf{\fee_{-i}}\right).
      \end{equation}

    \item \emph{Miner Utility:}
    \begin{equation}
    \label{eq:miner_util}
    \begin{split}
    u_{miner} (\mathbf{\fee} ; \mathbf{v})
    &
    \define
    \sum_{i=1}^{n'} \left(\pay(b_i, \mathbf{b}_{-i}) - \burn(b_i, \mathbf{b}_{-i}) \right) + \sum_{i=n+1}^{n'} u_i(\fee_i, \mathbf{\fee_{-i}} ; 0)
    \\&
    =
    \sum_{i=1}^n \left(\pay(b_i, \mathbf{b}_{-i}) - \burn(b_i, \mathbf{b}_{-i}) \right) - \sum_{i=n+1}^{n'} \burn(b_i, \mathbf{b}_{-i}).
    \end{split}
    \end{equation}

    \end{itemize}
\end{definition}

It will sometimes be convenient to use the shorthand 
$$u_C(X ; Y) \stackrel{\text{def}}{=} \sum_{i\in C} u_i(X ; Y).$$

We assume an auction is individually rational: For any $\mathbf{b} \in R_+^n$ and for each bidder $i$,
    \begin{equation}
    \label{eq:IR_util}
    u_i(b_i, \mathbf{b}_{-i} ; b_i) \geq 0.
    \end{equation}
    We also assume it is burn-balanced, i.e., the burn of each bidder does not exceed their payment: \begin{equation}
    \label{eq:BB}
    \pay(b_i, \mathbf{b}_{-i}) \geq \burn(b_i, \mathbf{b}_{-i}) \geq 0. 
    \end{equation}

    We assume that an auction is anonymous, so that a permutation over the bids results in a permutation over the confirmed bids. For any bid vector $\mathbf{b} \in R_+^n$, permutation $\pi \in S_n$, and rule $x \in \{\alloc, \pay, \burn\}$, then
    $x(\pi(\mathbf{b})) = \pi(x(\mathbf{b}))$.

    We say a bidder is a zero-utility bidder if $u_i(b_i, \mathbf{b_{-i}} ; b_i) = 0$: I.e., it is a bidder that is either unconfirmed, or pays their bid. 

We make the following tie-breaking assumption:

\begin{definition}
    \textit{Consistent Tie-Breaking towards zero-utility bidders, w.r.t. unconfirmed bidders.} 

    Consider a bidder $j$ and bids $\mathbf{b}_{-j}$ so bidder $j$ is unconfirmed both when bidding $b_j$, and when bidding $b'_j$. Call the former setting $A$ and the latter $B$. 

    Then any bidder $i$ that is zero-utility in both settings $A$ and $B$, has the same confirmation status in both.

\end{definition}

This definition is natural in the sense that it holds for reasonable auction formats. 

\begin{example}
Consider a single-item \textit{posted-price} auction with price $p$ and lexicographic tie-breaking, i.e., for bidders $b_1, \ldots, b_n$, the lowest index bidder with $b_i \geq p$ wins the item and pays $p$. Then this tie-breaking is consistent towards zero-utility bidders, w.r.t. unconfirmed bidders. Let us unpack the definition directly: Consider some unconfirmed bidder $j$. It either has a higher index than the confirmed bidder, or $b_j < p$. Consider it changes its bid and remains unconfirmed (then, in the latter case, we have $b'_j < p$). Now consider a zero-utility bidder $i$: If $i$ is unconfirmed, then $i$ remains unconfirmed. If $i$ is confirmed, then it must be that $b_i = p$, and it remains confirmed in both settings. 

Alternatively, consider a second-price auction with lexicographic tie-breaking. I.e., the highest bidder wins, and if there are several equal highest bids, the lowest indexed one wins. Similarly, this is consistent towards zero-utility bidders, w.r.t. unconfirmed bidders. 

\end{example}

We note that the property holds for specially-constructed formats, such as the discount auction of \cite{chung2024collusion}. 


\begin{definition}[$c$-\Glsxtrfull{SCP}, with Active / Passive Miner]
    \label{def:Scp}

    For a bid vector $b_1, \ldots, b_n$, a side-contract between a miner and a coalition of bidders $C$ is such that bidders in the coalition $i\in C$ may change their bids to $b'_i$. In the active model, the miner may also omit bids (whether by coalition members or not), and add $n' - n$ fake bids. We denote the original honest setting $A$ and the setting adjusted by the collusion $B$. We use the term \textbf{setting} to mean a certain combination of bids. We say that a collusion is beneficial if the joint utility of the miner and the bidders in $C$ is greater in $B$ than in $A$, given that their true valuations are from $A$. 

    \begin{equation}
        \label{eq:scp_cond} 
        u_{miner}(A ; A) + \sum_{i \in C} u_i(A ; A) < u_{miner}(B ; A) + \sum_{i \in C} u_i(B ; A).
    \end{equation}
    

    A TFM is $c$-\gls{SCP}-proof if a coalition of a miner and a coalition $C$ of the bidders with $|C| \leq c$ cannot increase their aggregate utility by deviating from the honest protocol.
    %
    
    %
\end{definition}

Notice that in the passive miner model, since the miner does not omit or introduce fake bids, we may denote $u_{miner}(X)$ instead of $u_{miner}(X ; Y)$, as $Y$ only matters in the miner's utility to distinguish the fake bids. We sometimes refer to a possible arrangement between a coalition and a miner as a collusion, distinguishing it from a \textit{beneficial} collusion, which satisfies the above condition that their joint utility improves. We also say there exists a $c$-SC ($c$-Side-Contract) when there is a beneficial collusion with $c$ bidders. 

\section{The class of \texorpdfstring{$2$}{}-SCP mechanisms equals the class of SCP mechanisms}

As observed by \cite{welfareIncreasingCollusion}, the Discount auction of \cite{chung2024collusion} shows a nice separation between $1$-SCP and $2$-SCP. 
We show that there are no other gaps between $c$-SCP and $(c+1)$-SCP mechanisms. We start by noting that we can restrict ourselves to the passive miner model of collusion. 

\begin{restatable}[]{lemma}{ActivePassiveMiner}
\label{lem:active_passive_miner}
    If a beneficial collusion exists against a mechanism in the active miner model, then it is either not $1$-SCP, or there exists a beneficial collusion against it in the passive miner model. 
\end{restatable}


\begin{lemma}[$2$-SCP $\implies$ a subset of the highest bids is confirmed]
\label{lem:highest_bidders}
If $a,p,\beta$ is $2$-SCP, then for a bids vector $b_1, \ldots, b_n$ sorted in descending order, it must be that $\exists 0 \leq \ell \leq n$ so that all bids up to and including $\ell$ are confirmed, and all other bids are not.
\end{lemma}
\begin{proof}
    Consider if for some two bids $b_i > b_j$, the latter is confirmed while the former is not, and consider the $2$-SC where bidder $i$ bids $b_j$ and bidder $j$ bids $b_i$. The miner utility is unchanged, the sum of utilities for agents $i,j$ increases by $b_i - b_j$, so it is a $2$-SC.
\end{proof}

\begin{lemma}[Canonical Ordering]
\label{lem:canonical_ordering}
    Consider a beneficial collusion of a miner and a coalition $C$ so that each $i\in C$ changes their bid from $b_i$ to $b'_i$ (possibly, with $b_i = b'_i$). Then, we can rearrange the collusion so that the highest bid among $\{b_i\}_{i\in C}$ goes to the highest bid among $\{b'_i\}_{i\in C}$, and so on, and this is a beneficial collusion.
\end{lemma}

\begin{proof}
The new collusion we define has the same payments and burns, since the mechanism is oblivious to the underlying values before the collusion. Moreover, by Lemma~\ref{lem:highest_bidders}, the top $\ell$ bids among $b'_i$ (for some $\ell$) are confirmed. The value after the collusion is maximized if the underlying values before the collusion in these confirmed spots are the top $\ell$ values in $b_i$, which is the case after the rearrangement. 
\end{proof}

\begin{figure}
\centering
\begin{minipage}{0.4\textwidth}
\includegraphics[width=\linewidth]{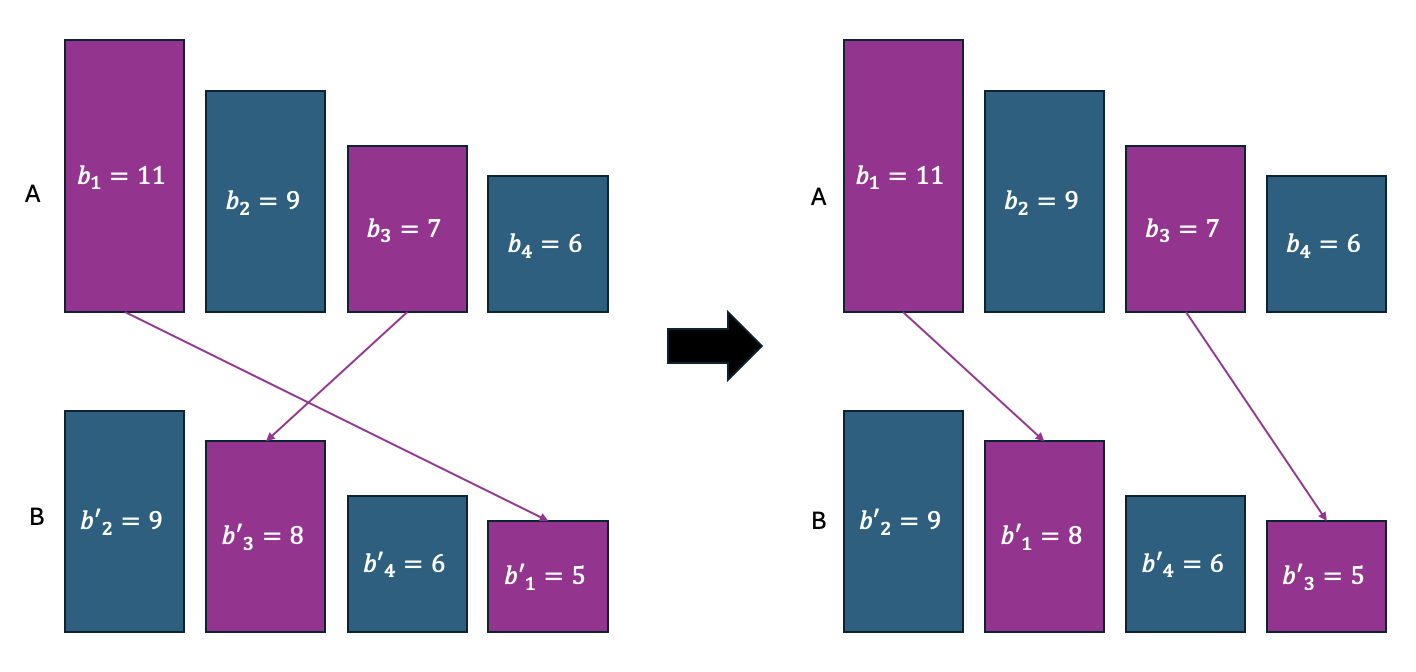}
\label{fig:canonical_ordering}
\end{minipage}  
\caption{A demonstration of Lemma~\ref{lem:canonical_ordering}. Given a collusion where bidders change their order after the collusion, we can find a collusion where the order is maintained.}
\Description[A demonstration of Lemma~\ref{lem:canonical_ordering}. Given a collusion where bidders change their order after the collusion, we can find a collusion where the order is maintained.]{A demonstration of Lemma~\ref{lem:canonical_ordering}. Given a collusion where bidders change their order after the collusion, we can find a collusion where the order is maintained.}
\end{figure}


\begin{lemma}[$2$-SCP $\implies$ A Form of Non-Bossiness]
\label{lem:2scp-non-bossy}
Consider a bidder $i$ bidding $b_i$ in $A$ and $b'_i$ in $B$, and this is the only change in bids between the settings. Assume bidder $i$ is unconfirmed in both settings. 
Then, the number of \textit{non} zero-utility bidders in $B$ is the same as in $A$. 
\end{lemma}
\begin{proof}

Assume towards contradiction there is a bidder $j$ that is a zero-utility bidder in $A$ and a non-zero-utility bidder in $B$. 

First, consider a collusion between the miner and bidder $i$. Then, by the $1$-SCP condition from $A$ to $B$ and vice versa,

\[
\begin{split}
    & u_{miner}(A) = u_{miner}(A) + u_i(A ; A) \\
    & \geq u_{miner}(B) + u_i(B ; A) = u_{miner}(B), \\
    & u_{miner}(B) = u_{miner}(B) + u_i(B ; B) \\
    & \geq u_{miner}(A) + u_i(A ; B) = u_{miner}(A), 
\end{split}
\]

which yields
$u_{miner}(A) = u_{miner}(B)$. 

Next, consider a collusion between the miner and bidders $i,j$. Then, by the $2$-SCP conditions from $A$ to $B$ and vice versa,

\[
\begin{split}
    & u_{miner}(A) = u_{miner}(A) + u_i(A ; A) + u_j(A ; A) \\
    & \geq u_{miner}(B) + u_i(B ; A) + u_j(B ; A), \\
    & u_{miner}(B) = u_{miner}(B) + u_i(B ; B) + u_j(B ; B) \\
    & \geq u_{miner}(A) + u_i(A ; B) + u_j(A ; B), 
\end{split}
\]

which yields:

\begin{equation}
\label{eq:2scp_baseline}
\begin{split}
& u_{miner}(A) \geq u_{miner}(A) + u_j(A ; B) + u_j(B ; A),
\end{split}
\end{equation}

i.e., 
$$u_j(A ; B) + u_j(B ; A) \leq 0.$$

However, notice that for bidder $j$, since it does not change its bid between $A$ and $B$, and it is a zero-utility bidder in $A$ and non-zero-utility bidder in $B$, it holds that:
$$u_j(A ; B) + u_j(B ; A) = 0 + b_j - p_i^B > 0,$$

in contradiction. 
\end{proof}


\begin{lemma}[Increase Monotonicity]
\label{lem:increase_monotonicity}
    If a collusion has a single unconfirmed bidder $\ell$ that increases their bid from $A$ to $B$, so that it remains lower than the lowest confirmed bidder in $A$, then either:

    (i) Bidder $\ell$ is confirmed in $B$. 

    (ii) The set of confirmed bidders weakly increases from $A$ to $B$. 



\end{lemma}
\begin{proof}
    Consider a collusion of the miner and all agents in a coalition $C$, and some single bidder $\ell\in C$ so that $b'_{\ell} > b_{\ell}$. Let $A$ be the setting before the collusion (i.e. with all agents bidding $\mathbf{b}$) and $B$ the setting after the collusion. 

If bidder $\ell$ is a 
non-zero-utility bidder at $B$, condition (i) is satisfied. 

Assume towards contradiction that bidder $\ell$ is a non-zero-utility bidder at $A$ but is a zero-utility bidder at $B$. 
By the $1$-SCP condition for the miner and agent $\ell$, w.r.t. going from $A$ to $B$ and vice versa, it holds that
$$u_{miner}(A) + b_{\ell}-p_{\ell}^A \geq u_{miner}(B) \geq u_{miner}(A) + b'_{\ell} - p_{\ell}^A,$$
which is a contradiction as $b'_{\ell} > b_{\ell}$. 

If bidder $\ell$ is a zero-utility bidder at both $A$ and $B$:
\begin{itemize}
    \item If bidder $\ell$ is confirmed in  $B$, then condition (i) is satisfied.

    \item If bidder $\ell$ is unconfirmed in both $A$ and $B$, then by Lemma~\ref{lem:2scp-non-bossy} the set of non-zero-utility bidders remains the same at $A$ and $B$. By consistent tie-breaking, so do zero-utility bidders, and so condition (ii) is satisfied.  

    \item If bidder $\ell$ is confirmed in $A$ and unconfirmed in $B$, consider the $1$-SCP condition from $A$ to $B$:

        \begin{equation}
        \label{eq:1scp_first}
    \begin{split}
        & u_{miner}(A) = u_{miner}(A) + u_i(A ; A) \geq u_{miner}(B) + u_i(B ; A), \\
        \end{split}
        \end{equation}
        Then, combine it with the $1$-SCP condition from $B$ to $A$:
        \[
        \begin{split}
        & u_{miner}(B) = u_{miner}(B) + u_i(B ; B) \geq u_{miner}(A) + u_i(A ; B) \\
        & \stackrel{\text{Eq.~\ref{eq:1scp_first}}}{\geq} u_{miner}(B) + u_i(B ; A) + u_i(A ; B) = u_{miner}(B) + b'_i - p_i^A \\
        & > u_{miner}(B) + b_i - p_i^A = u_{miner}(B),
    \end{split}
    \]

    which is a contradiction. 
\end{itemize}
\end{proof}

A similar (but opposing) statement holds for decreasing bids. 

\begin{lemma}[Decrease Monotonicity]
\label{lem:decrease_monotonicity}

    If a collusion has a single \textit{confirmed} bidder $\ell$ that decreases their bid from $A$ to $B$, so that it remains higher than the highest unconfirmed bidder in $A$, then either:

    (i) Bidder $\ell$ is unconfirmed in $B$. 

    (ii) The set of unconfirmed bidders weakly increases from $A$ to $B$. 

\end{lemma}

\begin{definition}[The ``Salsa Dance'' Decomposition.]
\label{def:car_parking}
Let $A$ and $B$ be two bid vectors
. Denote $U$ the bidders increasing their bid from $A$ to $B$, and $D$ the bidders decreasing their bids. Let $W_B$ be the confirmed bidders in $B$.  Let $U_I = U \cap W_B, U_O = U \cap L_B. $ Let $D_I = D \cap W_B, D_O = D \setminus W_B$. 

We construct $SEQ$ in the following way:

1. $SEQ = \{A\}$.

2. For each bidder $i\in U_I$ in descending order of their bids in $A$, let $X$ be the current last vector added to $SEQ$, then add $SEQ = SEQ \cup \{\{b'_i, X_{-i}\}\}$. 

3. For each bidder $i\in D$ in increasing order of their bids in $A$, let $X$ be the current last vector added to $SEQ$, then add $SEQ = SEQ \cup \{\{b'_i, X_{-i}\}\}$. 

4. For each bidder $i\in U_O$ in descending order of their bids in $A$, let $X$ be the current last vector added to $SEQ$, then add $SEQ = SEQ \cup \{\{b'_i, X_{-i}\}\}$. 
    
\end{definition}

The following lemma shows that the canonical ordering is maintained throughout the decomposition.

\begin{lemma}[Order-Maintaining Decompositions]
\label{lem:one_by_one_decomposition}
    Consider a canonical ordering collusion of a miner and a coalition $C$ so that each $i\in C$ changes their bid from $b_i$ to $b'_i$. Let $U$ be the bidders that raise their bid, and $D$ the bidders that decrease their bid. 
    Consider any decomposition where the subsequence of changes of bidders in $U$ is in descending order (higher bidders are earlier in the sequence), and the subsequence of changes of bidders in $D$ is in ascending order. 
Then, every two subsequent bid vectors, as a (not necessarily beneficial) collusion of the miner and the coalition $C$, are in canonical ordering.
\end{lemma}
\begin{proof}
    Consider any step $j$ in which a bidder raises their bid from $b_i$ to $b'_i$. 
    
    For all bidders $k$ who are higher in the canonical ordering (and have $b_k > b_i$), they either raise their bid, decrease their bid, or maintain it. If they raise their bid, then they already raised it in a previous step. By the canonical ordering, it holds that $b'_k > b'_i$, and so this is maintained after raising to $b'_i$, as the agent already moved to value $b'_k$. If the agent maintains or decreases their bid, this means that $b_k \geq b'_k > b'_i$, and so the order is maintained after the change to $b'_i$. 

    For all bidders $k$ who are lower in the canonical ordering, we have $b'_i > b_i \geq b_k$, and so they remain lower after the raise.

    Now consider a bidder which decreases their bid from $b_i$ to $b'_i$. 

    For all bidders $k$ who are lower in the canonical ordering, if they decrease their bid, then they already did it before agent $i$. Since $b'_k \leq b'_i$ by the canonical ordering, then even after agent $i$ decreases to $b'_i$, the order is maintained. If agent $k$ maintains or raise their bid, then they either already changed to $b'_k$, or maintained $b_k$, in either case by the canonical ordering at $B$, $b'_i \geq b'_k \geq b_k$, so the order is maintained after the decrease. 

    For all bidders $k$ who are higher in the canonical ordering, decreasing from $b_i$ to $b'_i$ maintains this order, as it only further decreases bidder $i$'s value.
\end{proof}

\begin{figure}
\centering
\begin{minipage}{0.45\textwidth}
\includegraphics[width=\linewidth]{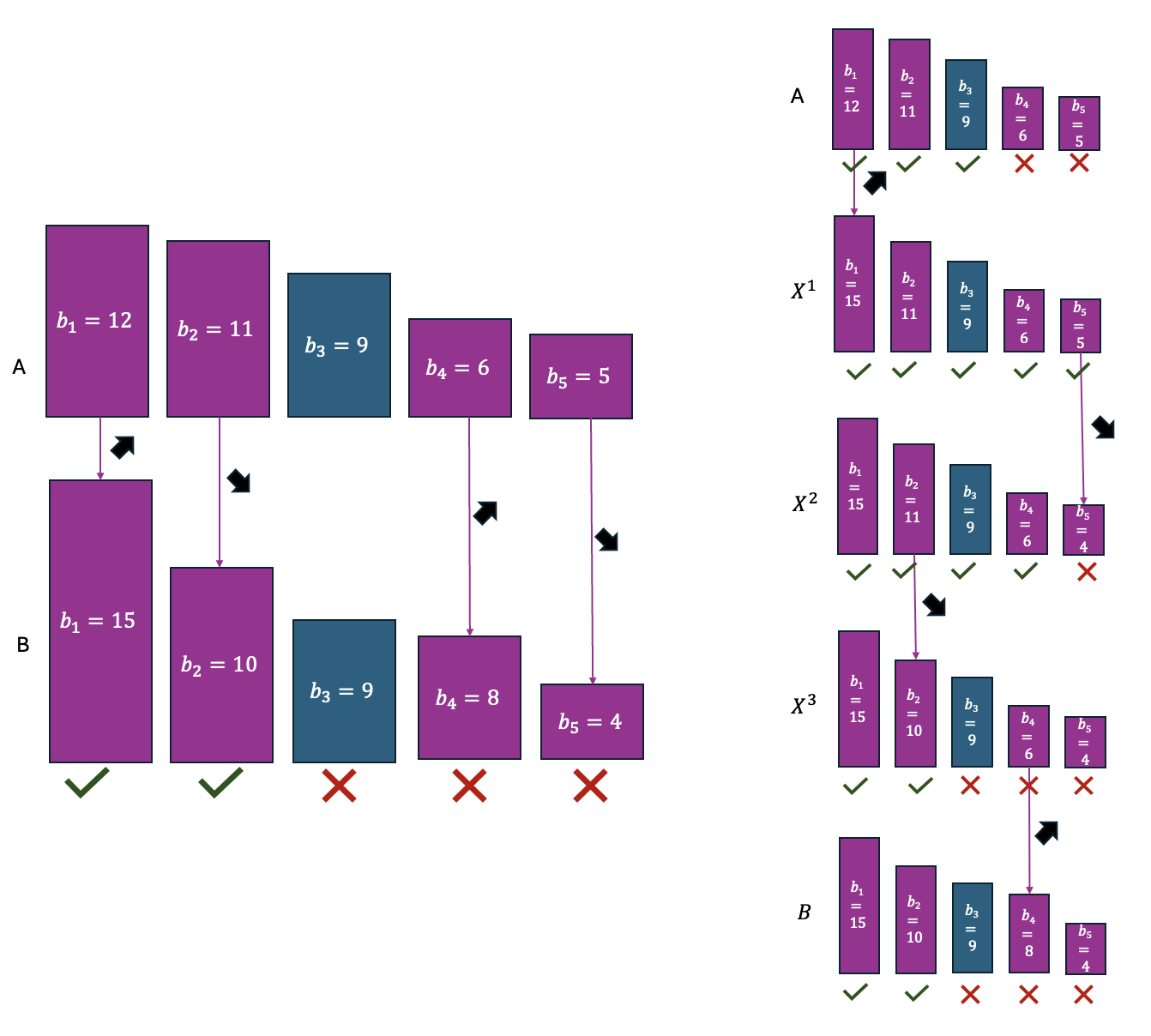}
\label{fig:car_parking_lemma}
\end{minipage} 
\caption{The ``Salsa Dance'' Lemma. We separate the bidders into $U_I, D, U_O$, and let each bidder change their bid one-by-one following Definition~\ref{def:car_parking}. Lemma~\ref{lem:car_parking} restricts the possible confirmations of each step in the decomposition, and the figure shows one such valid sequence of confirmations.}
\Description[The ``Salsa Dance'' Lemma. We separate the bidders into $U_I, D, U_O$, and let each bidder change their bid one-by-one following Definition~\ref{def:car_parking}. Lemma~\ref{lem:car_parking} restricts the possible confirmations of each step in the decomposition, and the figure shows one such valid sequence of confirmations.]{The ``Salsa Dance'' Lemma. We separate the bidders into $U_I, D, U_O$, and let each bidder change their bid one-by-one following Definition~\ref{def:car_parking}. Lemma~\ref{lem:car_parking} restricts the possible confirmations of each step in the decomposition, and the figure shows one such valid sequence of confirmations.}
\end{figure}

Before we prove the main structural lemma for the ``Salsa Dance'' decomposition, we demonstrate that the specific choice of decomposition matters for its success. 

\begin{example}
Consider an auction with $k=2$ items with the allocation rule that the two highest bidders $v_1, v_2$ are both confirmed if and only if $v_1+v_2 \geq2$. Regarding payments, $v_2$ pays their bid, and $v_1$ pays the bids’ average $\frac{v_1 + v_2}{2}$. All payments are burned. 
Consider bidders with true values $b_1 = 1, b_2 = 1$ that bid $b'_1 = \frac{5}{4}, b'_2 = \frac{3}{4}$. This is a $2$-SC that brings the joint utility of the two bidders and the miner from $0$ $\frac{1}{4}$. 
 Under the Salsa Decomposition, the two steps are $(1,1)\rightarrow (5/4,1)\rightarrow (5/4,3/4)$, and the second step is a $2$-SC (the joint utility changes from $\frac{1}{8}$ to $\frac{1}{4}$ (recall that for each decomposed step, we consider the beginning of the step to be the true values and the end of the step to be the collusion bids). 
 The other possible decomposition is $(1,1) \rightarrow (1,3/4) \rightarrow (5/4,3/4)$, where no step is a $2$-SC. 
\end{example}

\begin{lemma}[The ``Salsa Dance'' Lemma]
\label{lem:car_parking}
Every collusion from setting $A$ to $B$ has a one-by-one decomposition given by Definition~\ref{def:car_parking}, so that the following holds:
\begin{itemize}

    \item The initial setting is $A$ and the final setting is $B$, and in the decomposition each bidder only changes their bid once.

    \item At steps where bidders in $U_O$ change their bids, the bidder that changed their bid is unconfirmed after the step.  

    \item At steps where bidders in $D_I$ change their bids, the bidder that changed their bid is confirmed in the subsequent step.

\end{itemize}
\end{lemma}
\begin{proof}

    The first proposition holds by our explicit construction. 

    For the second proposition, examine the implications of the monotonicity Lemma~\ref{lem:increase_monotonicity}. If a bidder in $U_O$ is confirmed at some point during the stage where bidders in $U_O$ raise their bids, then the property that some bidder in $U_O$ is confirmed is maintained until the end of the process. This is true since at each step either Lemma~\ref{lem:increase_monotonicity}'s condition (i) or (ii) holds. If (i) holds, then the raising bidder is in $U_O$ and is confirmed. If (ii) holds, then the bidder in $U_O$ that was confirmed at the end of the step before remains confirmed. However, by the definition of $U_O$, no bidder in $U_O$ is confirmed at $B$, which is the last element of $SEQ$, so this is a contradiction.  

    This further implies that all bidders in $U_O$ are unconfirmed throughout the stage where bidders in $U_O$ raise their bids. By the non-bossy Lemma~\ref{lem:2scp-non-bossy} together with the consistent tie-breaking assumption, this means that the number of confirmed bidders remains the same throughout that stage, and so at the end of the decreasing stage by bidders in $D$, it is the same amount as in $B$. 

    For the third proposition, assume towards contradiction that a bidder in $D_I$ decreases their bid and is then unconfirmed. Similarly (in an opposing way) to the $U_O$ case, this implies that by the end of the decreasing stage there is a bidder in $D_I$ which is unconfirmed: At each step, either by condition (i) the decreasing bidder in $D_I$ is unconfirmed, or by (ii) a bidder in $D_I$ that is unconfirmed from the last step, remains so. However, we previously concluded that during the last stage where bidders in $U_O$ raise their bids, the number of confirmed bidders remains the same, in contradiction to all bidders in $D_I$ being non-confirmed in $B$. 
\end{proof}

\begin{lemma}[Single Mover Reduction]
\label{lem:single_mover}
    Given a one-by-one decomposition with the properties of Lemma~\ref{lem:car_parking}, there are two subsequent steps in the sequence so that there is a beneficial collusion from the preceding step to the subsequent step for the coalition $C$, where a single bidder changes their bid.  
\end{lemma}

\begin{proof}

    We assume towards contradiction that there are \textit{no} two subsequent steps that allow a beneficial collusion between them. We can write the condition for the miner and the coalition $C$ that prohibits the collusion for every two subsequent steps. We also write the condition of the original collusion (that gains a benefit of $\delta$). We show that summing all the inequalities, and the collusion equality with the term $\delta$, leads to a contradiction. 

    We introduce the following previously introduced notations: Let $A$ and $B$ be the settings of the collusion $\Omega_v$, and let $X^1, \ldots, X^k$ be all the intermediate steps of the decomposition $SEQ$. 

    When we argue for some bidder $i$, let $X^{i}_{before}$ be the setting before it changes its bid, and $X^i_{after}$ after. Since bidder $i$ is the only bidder that changes their bid in that step, it holds that $u_C(X^i_{after} ; X^i_{after}) - u_C(X^i_{after} ; X^i_{before}) = u_i(X^i_{after} ; X^i_{after}) - u_i(X^i_{after} ; X^i_{before})$. 
    We note the following: 

    \begin{itemize}
        \item For every bidder $i\in D_I$, by Lemma~\ref{lem:car_parking}, the bidder is confirmed in $X^i_{after}$. Thus, $$u_C(X^i_{after} ; X^i_{after}) - u_C(X^i_{after} ; X^i_{before}) = b'_i - b_i.$$ 

        \item For every bidder $i\in U_O$, by Lemma~\ref{lem:car_parking}, the bidder is unconfirmed in $X^i_{after}$, and thus
$$ u_C(X^i_{after} ; X^i_{after}) - u_C(X^i_{after} ; X^i_{before})  = 0.
$$

        \item For every bidder $i\in U_I$, it is either confirmed or not confirmed in $X^2$, and we have 
        \[
        \begin{split}
            & u_C(X^i_{after} ; X^i_{after}) - u_C(X^i_{after} ; X^i_{before})  \\
            & \leq \max \{ b'_i - p_i^{X^2} - (b_i - p_i^{X^2}), 0\} = b'_i - b_i.
        \end{split}
        \]

        \item For every bidder $i\in D_O$, it is either confirmed or not confirmed in $X^2$, and we have 
        \[
        \begin{split}
            & u_C(X^i_{after} ; X^i_{after}) - u_C(X^i_{after} ; X^i_{before}) \\
            & \leq \max \{ b'_i - p_i^{X^2} - (b_i - p_i^{X^2}), 0\} = 0.
        \end{split}
        \]
        
    \end{itemize}

    
    Now, we can write the SCP inequalities:
    \begin{equation}
\begin{split}
     & u_{miner}(A) + u_C(A ; A) \geq u_{miner}(X^1) + u_C(X^1 ; A), \\
    & u_{miner}(X^1) + u_C(X^1 ; X^1) \geq u_{miner}(X^2) + u_C(X^2 ; X^1), \\
    & \ldots \\
    & u_{miner}(X^k) + u_C(X^k ; X^k) \geq u_{miner}(B) + u_C(B ; X^k), \\
    & u_{miner}(B) + u_C(B ; A) = u_{miner}(A) + u_C(A ; A) + \delta. \\
     \end{split}
\end{equation}

First, notice that all $u_{miner}$ terms cancel out between the LHS and RHS when summing all the inequalities, and $u_C(A ; A)$ on both sides cancel as well. Let $i^*$ be the last agent in $U \cup D$ that changes their bid in the sequence (i.e., $X^{i^*}_{after} = X^k$). 
We thus get, when summing all the inequalities:

\begin{equation}
\begin{split}
    & u_C(B ; A) - u_C(B ; X_{after}^{i^*}) \\
    & + \sum_{i\in U \cup D \setminus \{i^*\}} u_C(X_{after}^i ; X_{after}^i) - u_C(X_{after}^i ; X_{before}^i) \geq \delta,
    \end{split}
\end{equation}


Then, using our conclusions for each bidder type 
, we get:

\begin{equation}
\label{eq:using_conclusions}
\begin{split}
    & u_C(B ; A) - u_C(B ; X_{after}^{i^*}) + \sum_{i\in (U_I \cup D_I) \setminus \{i^*\}} (b'_i - b_i) + \sum_{i\in D_O \cup U_O } 0 \\
    & \geq u_C(B ; A) - u_C(B ; X_{after}^{i^*})  \\
    & + \sum_{i\in U \cup D \setminus \{i^*\}} u_C(X_{after}^i ; X_{after}^i) - u_C(X_{after}^i ; X_{before}^i) \geq \delta.
    \end{split}
\end{equation}

Let $\sigma = U_I \cup D_I$. 
Since exactly the bidders in $\sigma$ are confirmed in $B$, we have:


\begin{equation}
\begin{split}
&     u_C(B ; A) = \sum_{i\in \sigma} (b_i - p_i^B), \\
&     u_C(B ;  X_{after}^{i^*}) = \sum_{i\in \sigma \setminus \{i^*\}} (b'_i - p_i^B) + 1[i^* \in \sigma] \cdot (b_{i^*} - p_{i^*}^B)
\end{split}
\end{equation}


Thus, we can continue Eq.~\ref{eq:using_conclusions} with:
\begin{equation}
\label{eq:final_form}
\begin{split}
    & 0 \stackrel{(*)}{=} \sum_{i\in \sigma} (b_i - p_i^B) - \left( \sum_{i\in \sigma \setminus \{i^*\}} (b'_i - p_i^B) + 1[i^* \in \sigma] \cdot (b_{i^*} - p_{i^*}^B) \right) \\
    & + \sum_{i\in \sigma \setminus \{i^*\}} (b'_i - b_i) \\
    & = u_C(B ; A) - u_C(B ; X_{after}^{i^*}) + \sum_{i\in \sigma \setminus \{i^*\}} (b'_i - b_i) + \sum_{i\in \sigma } 0 \geq \delta.
    \end{split}
\end{equation}

The proof of $(*)$ is a matter of case-analysis and we show it in the appendix. 
We conclude that Eq.~\ref{eq:final_form} shows $0 \geq \delta$, but since $\delta > 0$, this is a contradiction, and there must be some two subsequent steps constituting a beneficial collusion. 
\end{proof}

\begin{restatable}[Single-Jump Property of Single Mover Collusion Coalition Value]{lemma}{SingleMoverIncrease}
\label{lem:single_mover_increase}

    In a $2$-SCP mechanism with a collusion where a single bidder $i$ increases their bid from $b_i$ to $b'_i$ and a coalition $C$ gains $\Delta$ in value, there is a single value $b_i^*$, so that for any small enough $\epsilon > 0$, there is a collusion for coalition $C$ where only bidder $i$ increases their value from $b_i^* - \frac{\epsilon}{2}$ to $b_i^* + \frac{\epsilon}{2}$, and the coalition gains $\Delta$ in value. 
\end{restatable}

A similar statement holds when the single mover decreases their value:

\begin{lemma}
\label{lem:single_mover_decrease}
    In a $2$-SCP mechanism with a collusion where a single bidder $i$ decreases their bid from $b_i$ to $b'_i$ and a coalition $C$ gains $\Delta$ in value, there is a single value $b_i^*$, so that for any small enough $\epsilon > 0$, there is a collusion for coalition $C$ where only bidder $i$ decreases their value from $b_i^* + \frac{\epsilon}{2}$ to $b_i^* - \frac{\epsilon}{2}$, and the coalition gains $\Delta$ in value. 
\end{lemma}

Finally, we can reduce any SCP to one with a single mover, and show by the single-jump property that we can gain a fixed value of $\Delta$ for the coalition while losing at most $\epsilon$ for the moving bidder, and can further isolate a single bidder, possibly different than the mover, that gains and allows a $2$-SCP. 

\begin{restatable}[One Mover, One Beneficiary]{theorem}{SCPTwoSCP}
\label{thm:scp_2scp}
    The class of $2$-SCP mechanisms equals the class of SCP mechanisms.
\end{restatable}

\section{Corollaries Of The Main Characterization}

We present two applications of the above characterization. 

\begin{lemma}
\label{lem:single_item_char}
    With a block size of $1$ (equivalently, a single-item auction), the class of $2$-SCP TFMs is exactly the class of first-price auctions with reserve $r$ that is burned. 
    \end{lemma}
\begin{proof}

    ($2$-SCP $\implies$ first-price with burned reserve)
    
    Consider any $2$-SCP auction $a, p, \beta$. By our characterization, it is also SCP. \cite{welfareIncreasingCollusion} show that SCP implies Global-SCP, and provide a characterization of Global-SCP auctions in the single-item case (notice that this means SCP implies the characterization, and we do not need to worry about Global-SCP, which we did not define). The characterization is as follows: For some reserve price $r$, the auction confirms the item to the highest bidder above $r$, and burns $r$. The payment function may be determined arbitrarily (as long as it respects individual rationality and burn-balance). 

    Now, consider the payment function for our auction, which is not only Global-SCP, but SCP. The function can not depend on losing bids: If it does, there is an SCP of the losing bids to increase the payment, which benefits the miner (and does not harm the losing bids). Thus, there is $f: R_+ \rightarrow R_+$ so that $p(\mathbf{b}) = f(\max \mathbf{b})$, i.e., the payment only depends on the winning bid. Now, assume towards contradiction that for some $v$, $f(v) < v$. Consider some $f(v) < v' < v$, and let $\mathbf{b} = (v, v')$ (call it setting $A$). Consider the $1$-SCP between the miner and bidder $2$ where bidder $2$ bids $2 \cdot v$ (call it setting $B$). We have $u_{miner}(A) + u_2(A ; A) = f(v) - r$, where $r$ is the reserve price that is burned. We also have $u_{miner}(B) + u_2(B ; A) = v' - r > f(v) - r$, and so this is a $1$-SC. We conclude that we must have $f(v) = v$. 

    (first-price with burned reserve $\implies$ SCP)

    If all bidders bid honestly, we have $$u_{miner}(\mathbf{b}) = 1[\max \mathbf{b} \geq r] \cdot (\max \mathbf{b} - r),$$ and $u_i(\mathbf{b}) = 0$ for any bidder $i$. Under any collusion, if bidder $i \in C$ ends up the winner, the joint utility of the collusion is $b_i - r$. If no bidder in the coalition wins, the coalition utility is either $p_i - r$ for some $i\not \in C$, or $0$. In any case, this is less than $1[\max \mathbf{b} \geq r] \cdot (\max \mathbf{b} - r)$, and so the SCP condition holds. 

\end{proof}


This characterization is interesting as it shows that the current popular implementations (Bitcoin's pay-as-bid auction, and Ethereum's EIP-1559, which is essentially a first-price auction with burned reserve) fit exactly the class of collusion-free auctions, in the SCP sense. 
Moreover, in the single-item case, the $2$-SCP characterization can be tightened to $1$-SCP. 

\begin{lemma}
\label{lem:single_item_1scp_2scp}
    With a block size $1$, $1$-SCP $\implies$ $2$-SCP.
\end{lemma}
\begin{proof}
    We prove the contrapositive, i.e., we show that if $2$-SCP is violated, then so is $1$-SCP. Consider a beneficial $2$-SCP collusion $\Omega$, going from setting $A$ to $B$. By Lemma~\ref{lem:single_mover}, we can assume there is a single bidder $i$ who changes their bid. If the other bidder in the collusion $j$ does not win the item under $B$, then their utility in $B$ is $0$, and by individual rationality of $A$, they do not gain from the collusion. Thus, the $1$-SC that only involves the miner and agent $i$ is beneficial as well. 

    Otherwise, the winner in $B$ is bidder $j$. Then, the collusion can be implemented without the involvement of bidder $i$: Instead of bidder $i$ changing their bid (from $b_i^A$ to $b_i^B$), the miner drops their bid $b_i^A$, and introduces a fake bid $b_i^B$. Call this collusion, that involves only bidder $j$, $\Omega'$. Now, notice that the utility of the miner and bidder $j$ under $B$ with is the same as the utility of the miner and bidders $i,j$ under $B$ (since $i$ has utility $0$). The utility of the miner and bidder $j$ under $A$ is, by individual rationality, at most the utility of the miner and the bidder $i,j$ under $A$. Thus, if $\Omega$ is beneficial for the colluders, then so is $\Omega'$. 
\end{proof}

The corollary immediately follows from the two lemmas:

\begin{corollary}
\label{corr:single_item_char}
      With a block size of $1$ (equivalently, a single-item auction), the class of $1$-SCP TFMs is exactly the class of first-price auctions with reserve $r$ that is burned.  
\end{corollary}


In \cite{chung2023foundations}, they show an impossibility for UIC + $1$-SCP mechanisms with a finite block size. However, with an infinite block size, \cite{chung2024collusion} present the discount auction, which is UIC+$1$-SCP. We complement these results by characterizing UIC+$2$-SCP mechanisms in the infinite block-size case:

\begin{definition}
    \textit{Fully-burned Posted-price Mechanisms} are the class of mechanisms where for some constant $r$, all bidders with $b_i \geq r$ are confirmed, they each pay $r$, and all payments are burned.  
\end{definition}

\begin{restatable}[]{theorem}{FullyBurnedPostedPrices}
\label{thm:fully_burned_posted_price}
    The class of UIC + $2$-SCP mechanism for infinite block size is exactly the class of fully burned posted-price mechanisms.  
    
\end{restatable}

\section{Discussion}

Our work shows that the class of $2$-SCP mechanisms is the same as the class of SCP mechanisms. This, in the view of the miner, means that when a mechanism is subject to collusion, they only need to find $2$ co-conspirators to implement a beneficial one. This shows that in some sense of \textit{coalitional} complexity, collusion is simple. 

One non-trivial assumption we make is that of consistent tie-breaking towards \textit{zero-utility} bidders w.r.t. unconfirmed bidders. Let us describe an auction which violates this assumption, and for which our single-mover argument fails.
\begin{example}
    Consider an auction where payments are fully burned. A winning bid is only confirmed if there is a losing bid with value at least $8$. If a bid with value at least $10$ exists, all bids of value $8$ and more are confirmed and pay $6.5$. 

    Consider two bidders $b_1 = 10, b_2 = 1$ (call this setting $A$), that collude to change their bids to $b'_1 = 9, b'_2 = 8$ (call this setting $B$). This is a beneficial collusion of the miner and two bidders as their joint utility increases from $0$ to $1$. However, no single-move decomposition of the collusion is beneficial. Let $X$ be the setting with bids $9, 1$, and $Y$ the setting with bids $10, 8$. The joint utility of the coalition moving from $A$ to $Y$ goes from $0$ to $-2$. From $A$ to $X$ it goes from $0$ to $0$. From $X$ to $B$ it goes from $0$ to $0$, and from $Y$ to $B$ it goes from $5$ to $1$. 
\end{example}

Going beyond our discussion of collusion size, there are other aspects of the collusion that may still be hard. In terms of \textit{computational} complexity, it may be hard to identify how to implement a successful collusion. This depends on the computational model of how the problem is represented. Importantly, at least at surface-level, it seems like our result does not change the basic complexity. Consider two natural choices. If we assume that the problem is given as a full discrete table mapping all possible types (values) to outcomes (confirmation, payment, and burn), then the problem of finding a collusion is a matter of iterating over all possible types, and finding two settings where a beneficial collusion exists. This is polynomial (quadratic) in the input size, and true even without our reduction. On the other hand, if one assumes as in \cite{golowichLi2022} that the problem is given as a Boolean circuit, then the problem of finding a beneficial $2$-SCP is coNP-complete. 

\begin{definition}
    \textit{Boolean circuit representation} of $a, p, \beta$: Consider a finite set of values $v^1, \ldots, v^k$, and let $\bar{v} \in \{0,1\}^{\log k}$ be binary representation of them. With $n$ bidders, let $C^a_1, \ldots, C^a_n$ be Boolean circuits that receive $n$ boolean vectors of size $\log k$ (so overall $n \log k$ bits), where each $C^a_i$ outputs $a_i(v_1, \ldots, v_n) \in \{0,1\}$. $C^p_1, \ldots, C^p_n, C^{\beta}_1, \ldots, C^{\beta}_n$ similarly each output a binary representation of the payment and burn for each bidder (which we assume, for simplicity, are in $v^1, \ldots, v^k$ as well). 
\end{definition}

We are interested in the search problem of finding a $2$-SC (or reporting there is none, if none exists). To understand its complexity, we can define the decision problem of whether the auction is $2$-SCP:

\begin{definition}
    $2$-SCPDP ($2$-SCP Decision Problem): Given circuits $C^a_1, \ldots, C^a_n, C^p_1, \ldots, C^p_n, C^{\beta}_1, C^{\beta}_n$ that represent an individually rational, burn-balanced single-item auction, decide whether the auction represented by the circuits is $2$-SCP.
\end{definition}

The problem is in coNP, since a witness to show that it is \textit{not} $2$-SCP is simply two sets of bidder values $v_1, \ldots, v_n, v'_1, \ldots, v'_n$, and (at most) two bidder indices $i,j$. Given this polynomial-size witness, a beneficial $2$-SC can be verified in polynomial time. We next define Boolean tautology, which is coNP-complete. 

\begin{definition}
    BOOLEAN TAUTOLOGY: Given a Boolean circuit $c$ with $n$ input bits, decide whether the circuit always outputs $1$.  
\end{definition}


\begin{restatable}[]{lemma}{coNPComplete}
\label{lem:coNPcomplete}
    $2$-SCPDP is coNP-complete.
\end{restatable}

In terms of \textit{communication} complexity, it is important to understand how many rounds of communication are needed between the colluders to establish the collusion, given that the values of the bidders are private information. This is especially important when the collusion coordination must be very quick, as Ethereum mines a block every $12$ seconds on average \cite{ethereumLatency}. Ideally (for the colluders), no communication would be needed, and a \textit{tacit} collusion strategy exists that each may apply on their own. It is clear that at most $2$ additional rounds of communication are required for a collusion: One round where all bidders share their values with the miner, another one where the miner tells each bidder what value to bid, and then the mandatory round where all bidders make their bids. It is interesting to understand in what settings the first two rounds can be avoided. Moreover, this suggestion leaves all the bargaining power in the hands of the miner, and this may not be acceptable for the bidders. It is interesting to know the ``price of bargaining'' in terms of communication, if one was to instead follow an established process of bargaining \cite{rubinsteinBargaining, compteJehiel}. 

\subsection{Open questions.}

We list some remaining open questions. 

\begin{itemize}
\item 
\textit{Removing the consistent-tie-breaking assumption.} Although we find the assumption natural, we would ideally wish to avoid limiting the range of auction formats we consider. Moreover, it remains possible that a non consistent-tie-breaking auction \textit{does} separate $2$-SCP from higher $c$-SCP, and this could be interesting for collusion prevention. 

    \item \textit{A separation of $1$-SCP and $2$-SCP with a finite block size greater than $1$.} The discount auction \cite{chung2024collusion} shows a separation in the infinite block size case. Our Lemma~\ref{lem:single_item_1scp_2scp} shows they are the same with a block size of $1$. What happens in-between?

    \item \textit{A Global-SCP mechanism which is $1$-SCP, but not $2$-SCP.} This is an open question of \cite{welfareIncreasingCollusion} that remains open. Global-SCP is weaker than SCP, so this is not ruled out by our results. 

    \item \textit{Characterizing SCP and Global-SCP beyond auctions.} The TFM literature is focused on auctions, but the collusion notions are general and can apply to other mechanism design settings.

\end{itemize}


\bibliographystyle{ACM-Reference-Format}
\bibliography{main}

\appendix

\section{Missing Proofs}

\ActivePassiveMiner*

\begin{proof}
    Consider a beneficial collusion in the active miner model. It consists of bidders changing their bid, the miner omitting bids and adding fake bids. Consider that for every added fake bid $b'_j$, we instead consider a bidder with value $0$ that changes their bid to $b'_j$, and add it to the collusion. Then, the joint utility of the revised collusion is the same before and after, and is thus also beneficial, without the miner adding fake bids. Regarding omitting bids, consider the setting $X$ where the miner omits all the bids it omits from $A$, but the bidders do not change their bids. We claim that there is either a beneficial collusion of the miner and the coalition $C$ from $A$ to $X$, or from $X$ to $B$. Assume towards contradiction otherwise. Then, the SCP conditions yield:
    \begin{equation}
    \label{eq:sophies_choice}
    \begin{split}
        u_{miner}(A) + u_C(A ; A) \geq u_{miner}(X) + u_C(X ; A), \\
        u_{miner}(X) + u_C(X ; X) \geq u_{miner}(B) + u_C(B ; X).
    \end{split}
    \end{equation}

    However, notice that $u_C(X ; X) = u_C(X ; A), u_C(B ; X) = u_C(B ; A)$, since any confirmed bid would come out of bidders that were not omitted, and their values do not change. Thus, Eq.~\ref{eq:sophies_choice} yields $u_{miner}(A) + u_C(A ; A) \geq u_{miner}(B) + u_C(B ; A)$, which is in contradiction to there being a beneficial collusion from $A$ to $B$. Thus, at least one of the two inequalities of Eq.~\ref{eq:sophies_choice} must be violated.

    If the first inequality of Eq.~\ref{eq:sophies_choice} is violated, then there is a beneficial collusion where no bidder moves, and the miner omits some bids. If we have $u_{miner}(X) > u_{miner}(A)$, then let $i$ be the agent with maximal $u_i(X ; A) - u_i(A ; A)$, a collusion of the miner with this user is a $1$-SC. If $u_{miner}(X) \leq u_{miner}(A)$, then there is a collusion where the miner adds fake bidders to $X$ (rather than omit these bidders from $A$), and as noted we can transform this into a collusion where $0$ bidders change their bids rather than the miner manipulates. 

    If the second inequality of Eq.~\ref{eq:sophies_choice} is violated, then we have a passive miner beneficial collusion. 
\end{proof}

\begin{claim}
    Proof of (*) in Lemma~\ref{lem:single_mover}
\end{claim}

\begin{proof}
    If $i^* \in \sigma$, then 

\[
\begin{split}
    & \sum_{i\in \sigma} (b_i - p_i^B) - \left( \sum_{i\in \sigma \setminus \{i^*\}} (b'_i - p_i^B) + 1[i^* \in \sigma] \cdot (b_{i^*} - p_{i^*}^B) \right) \\
    & + \sum_{i\in \sigma \setminus \{i^*\}} (b'_i - b_i) \\
    & = \sum_{i\in \sigma} (b_i - p_i^B) - \left( \sum_{i\in \sigma \setminus \{i^*\}} (b'_i - p_i^B) + (b_{i^*} - p_{i^*}^B) \right) + \sum_{i\in \sigma \setminus \{i^*\}} (b'_i - b_i) \\
    & = \sum_{i\in \sigma} b_i - \left( \sum_{i\in \sigma \setminus \{i^*\}} b'_i  + b_{i^*} \right) + \sum_{i\in \sigma \setminus \{i^*\}} (b'_i - b_i) = 0,
\end{split}
\]

and if $i^* \not \in \sigma$, then
\[
\begin{split}
    & \sum_{i\in \sigma} (b_i - p_i^B) - \left( \sum_{i\in \sigma \setminus \{i^*\}} (b'_i - p_i^B) + 1[i^* \in \sigma] \cdot (b_{i^*} - p_{i^*}^B) \right) \\
    & + \sum_{i\in \sigma \setminus \{i^*\}} (b'_i - b_i) \\
    & = \sum_{i\in \sigma} (b_i - p_i^B) - \left( \sum_{i\in \sigma} (b'_i - p_i^B) \right) + \sum_{i\in \sigma } (b'_i - b_i) \\
    & = 0.
\end{split}
\]

\end{proof}

\SingleMoverIncrease*

\begin{proof}
    Let $A$ be the setting before the collusion and $B$ the setting after, and assume $a_i^B = 0$. Then, we write two inequalities for the $1$-SCP condition for the miner and agent $i$, when increasing from $b_i$ to $b'_i$, and vice versa.

\begin{equation}
    u_{miner}(B) \leq u_{miner}(A) + u_i(A ; A) = u_{miner}(A) + a_i^A \cdot (b_i - p_i^A),
\end{equation}

\begin{equation}
    u_{miner}(A) + a_i^A \cdot (b'_i - p_i^A) = u_{miner}(A) + u_i(A ; B) \leq u_{miner}(B). 
\end{equation}

Overall, this yields $a_i^A \cdot b'_i \leq a_i^A \cdot b_i$, which only holds when $a_i^A = 0$. I.e., it would violate $1$-SCP if $a_i^A = 1, a_i^B = 0$. For the case of $a_i^A = a_i^B$, the two inequalities yield $u_{miner}(A) = u_{miner}(B)$, and thus $u_{miner}(A) + u_i(A ; A) = u_{miner}(B) + u_i(B ; A)$. We show that we get the same equality if $a_i^B = a_i^A = 1$:

Consider if $a_i^B = a_i^A = 1$. We again write two inequalities for the $1$-SCP condition for the miner and agent $i$, when increasing from $b_i$ to $b'_i$, and vice versa. 

\begin{equation}
    u_{miner}(B) + b_i - p_i^B \leq  u_{miner}(A) + b_i - p_i^A,
\end{equation}

\begin{equation}
    u_{miner}(A) + b'_i - p_i^A \leq u_{miner}(B) + b'_i - p_i^B.
\end{equation}

This yields $u_{miner}(A) - p_i^A = u_{miner}(B) - p_i^B$ and thus $u_{miner}(A) + u_i(A ; A) = u_{miner}(B) + u_i(B ; A)$. 

Now, we show that given this condition, that we have shown to hold whenever $a_i^B = a_i^A$, we can find a $2$-SCP if the collusion condition for the entire collusion of the miner and the coalition $C$ holds:

\begin{equation}
    u_{miner}(B) + \sum_{i\in C} u_i(B ; A) >  u_{miner}(A) + \sum_{i\in C} u_i(A ; A), 
\end{equation}

yields:
$$\sum_{j \in C \setminus \{i\}} (u_j(B;A) - u_j(A ; A) ) >  0,$$

and so there must be some $j$ with $u_j(B;A) - u_j(A ; A) > 0$. If we consider the coalition of the miner, $i$ and $j$, we get a $2$-SCP. 

We then conclude that bidder $i$ is unconfirmed when bidding $b_i$, confirmed when bidding $b'_i$, and if we let $b_i^*$ be the critical bid so that the bidder becomes confirmed (which must exist by our previous discussion) then moving from $b_i$ to any bid lower than $b_i^*$, in particular, $b_i^* - \frac{\epsilon}{2}$, does not change the value of the coalition, and moving from $b_i^*+\frac{\epsilon}{2}$ to $b_i'$ also does not change the value of the coalition, and thus moving from $b_i^* - \frac{\epsilon}{2}$ to $b_i^*+\frac{\epsilon}{2}$ must increase the coalition value by $\Delta$. 

\end{proof}

\SCPTwoSCP*

\begin{proof}

    Every SCP mechanism is also $2$-SCP by logical inclusion. We thus focus on showing that if there is an SCP against a mechanism, there is also a $2$-SCP against the mechanism. We assume towards contradiction the contrary: There is an SCP against a mechanism, but the mechanism is $2$-SCP-proof. 

    By Lemma~\ref{lem:single_mover}, if there is any collusion against the mechanism, then we have a collusion where a single agent $i$ changes their bid.

Consider if $b_i < b'_i$ (the single mover increases their bid). By Lemma~\ref{lem:single_mover_increase}, we can assume that $b'_i - b_i \leq \epsilon$ for any $\epsilon > 0$, and $a_i^B = 1, a_i^A = 0$. By the two $1$-SCP inequalities (increasing from $b_i$ to $b'_i$ and vice versa) we have:

\begin{equation}
    u_{miner}(B) + u_i(B ; A) = u_{miner}(B) + b_i - p_i^A \leq u_{miner}(A),
\end{equation}

\begin{equation}
\begin{split}
    & u_{miner}(A) \leq u_{miner}(B) + b'_i - p_i^A \\
    & < u_{miner}(B) + b_i - p_i^A + \epsilon = u_{miner}(B) + u_i(B ; A) + \epsilon. 
    \end{split}
\end{equation}

Similarly, if $b_i > b'_i$ (the single mover decreases their bid), we can assume $b_i - b'_i \leq \epsilon$, and consider the $1$-SCP condition for the miner and agent $i$, when decreasing from $b_i$ to $b'_i$, where $b_i < b'_i + \epsilon$. Then,

    $$u_{miner}(B) \leq u_{miner}(A) + b_i - p_i^A, $$

$$ u_{miner}(A) + b'_i - p_i^A \leq u_{miner}(B),
$$

and so 

\[
\begin{split}
& u_{miner}(B) + u_i(B ; A) = u_{miner}(B) + p_i^A - b_i \leq u_{miner}(A) \\
& \leq u_{miner}(B) + p_i^A - b'_i  = u_{miner}(B) + p_i^A - b_i + \epsilon \\
& = u_{miner}(B) + u_i(B ; A) + \epsilon,
\end{split}
\]
and so in both cases we have:
\begin{equation}
\label{eq:eps_condition}
  u_{miner}(B) + u_i(B ; A) \leq u_{miner}(A) \leq   u_{miner}(B) + u_i(B ; A) + \epsilon. 
\end{equation}

Then, the collusion condition for the entire collusion of the miner and the coalition $C$ yields:

\begin{equation}
\begin{split}
    & u_{miner}(B) + \sum_{j\in C} u_j(B ; A) \geq u_{miner}(A) + \sum_{j\in C} u_j(A ; A) + \Delta \\
    & \geq u_{miner}(B) + u_i(B ; A) + \sum_{j\in C \setminus \{i\}} u_j(A ; A) + \Delta - \epsilon, 
    \end{split}
\end{equation}

where $\Delta$ is the original difference in the collusion before and after the value. If we eliminate $u_{miner}(B) + u_i(B ; A) + \sum_{j\in C \setminus \{i\}}u_j(A ; A)$ from both sides we get:

$$\sum_{j\in C \setminus \{i\}} \left( u_j(B ; A) - u_j(A ; A) \right) \geq \Delta - \epsilon > 2|C| \cdot \epsilon,$$

For a choice of $\epsilon < \frac{\Delta}{2|C| + 1}$, there must thus be some agent $j$ with $u_j(B ; A) - u_j(A ; A) \geq 2 \epsilon$, or the sum of utility differences will be too small to have a $\Delta$ difference in the coalition value. We get a $2$-SCP of the miner, agent $i$ and agent $j$:

\begin{equation}
\begin{split}
    & u_{miner}(B) + u_i(B ; A) + u_j(B ; A) \geq u_{miner}(A) + u_i(A ; A) + \epsilon + u_j(A ; A) \\
    & > u_{miner}(A) + u_i(A ; A) + u_j(A ; A).  
    \end{split}
\end{equation}
\end{proof}

\FullyBurnedPostedPrices*

\begin{proof}



(\emph{Fully Burned Posted-prices} $\implies$ UIC + $2$-SCP)
Posted-prices are UIC as they satisfy monotonicity and critical bid \cite{myerson1981optimal}. Burning the payments does not affect bidder incentives, and so this is true also when the payments are burned. The miner utility is always zero (since all payments are burned), and so the value of any coalition is the utility of its bidders. For any coalition $C$ of bidders, their utility is maximized when each of them is confirmed according to $1[b_i \geq r]$, and so it is maximized in the original mechanism without any collusion. 

(UIC + $2$-SCP $\implies$ \emph{Fully Burned Posted-prices})
We know that UIC + $2$-SCP $\stackrel{\text{Theorem~\ref{thm:scp_2scp}}}{\implies}$ UIC + SCP $\stackrel{\text{\cite{welfareIncreasingCollusion}}}{\implies}$ UIC + Global-SCP. Thus, the characterization of \cite{welfareIncreasingCollusion} (their Theorem 3.5) for UIC + Global-SCP mechanisms holds in our case. In this characterization, the amount of burn is a function of the number of confirmed bids $\beta^{cardinal}(k)$, and we can encode it as a function $\beta^{diff}(k)$ so that $\beta^{cardinal}(k) = \sum_{j=1}^k\beta^{diff}(k)$. The characterization requires $\beta^{avg}(k) \stackrel{def}{=} \frac{1}{k}\sum_{j=1}^k\beta^{diff}(k)$ to be an increasing function of $k$, and specifies the payments for the bidders. The confirmed bids are the maximal set of highest bids, so that each makes a marginal positive contribution to the joint utility (i.e., its value is higher than the additional burn). Now consider if for some $k$ and $k+1$, $\beta^{avg}(k)$ is a strictly increasing function (and constant for any smaller values of $k$). Let $\beta^{avg}(k) < \rho < \beta^{avg}(k+1)$, and consider $k+1$ bidders with value $\rho$. Then, exactly $k$ of the bidders are confirmed (as they have a positive marginal contribution to the joint utility up to the value $k$). By the payments given by the characterization, they each pay the minimum between the unconfirmed bid $\rho$, and some expression that is at least $\beta^{diff}(k+1) \geq \beta^{avg}(k+1) > \rho$, and thus they each pay $\rho$. Thus, the payments are $k \cdot \rho > k \cdot \beta^{avg}(k)$, which is the burn. However, \cite{chung2023foundations} show that any UIC+$1$-SCP TFM has zero miner revenue, and so we know all payments are burned, which is a contradiction. We conclude that $\beta^{avg}(k)$ is weakly increasing (by the characterization of \cite{welfareIncreasingCollusion}), but never strictly increasing, i.e., it is constant. As Theorem~3.9 of \cite{welfareIncreasingCollusion} shows, these are exactly the class of fully burned posted-price mechanisms. 

\end{proof}

\coNPComplete*

\begin{proof}
    We have already shown that $2$-SCPDP is in coNP. We show it is coNP-complete by a reduction from BOOLEAN TAUTOLOGY. Consider a circuit $C$ with $n$ input bits $q_1, \ldots, q_n$. 
    We construct circuits $C^1, \ldots, C^{n+2}$, each with $n+2$ input bits, in the following way: \[
    \begin{split}
    & C^1(1,0, q_1, \ldots, q_n) = C(q_1, \ldots, q_n), \\
    & C^1(0,1,q_1, \ldots, q_n) = 1 - C(q_1, \ldots, q_n),\\
    & C^1(1,1, q_1, \ldots, q_n) = 1, C^1(0,0,q_1, \ldots, q_n) = 0.
    \end{split}
    \]
    Let 
    \[
    \begin{split}
    & C^2(1,0, q_1, \ldots, q_n) = 1 - C(q_1, \ldots, q_n), \\
    & C^2(0,1,q_1, \ldots, q_n) = C(q_1, \ldots, q_n), \\
    & C^2(1,1, q_1, \ldots, q_n) = 0, C^2(0,0,q_1, \ldots, q_n) = 0.
    \end{split}
    \]
    For any $j > 2$, let 
    \[
    \begin{split}
    & C^j(s_1, s_2, q_1, \ldots, q_n) = 1[q_{j-2} > \max_{j' < j-2} \{q_{j'}, s_1, s_2\}].
    \end{split}
    \]
    
    We now use these auxiliary circuits to construct the $2$-SCPDP. Consider possible bidder values in $\{0, 1\}$, and $n+2$ bidders.   
    
    Let $C^a_j = C^j, C^p_j(b_1, \ldots, b_{n+2}) = C^j(b_1, \ldots, b_{n+2}) \cdot b_j, C^{\beta}_j = 0$. We claim that this construction always outputs YES for $2$-SCPDP when $C$ is a YES instance, and NO when it is a NO instance. 
    
    Notice that, if $C$ always outputs $1$, then our construction is a single-item first-price auction: Bidder $1$, whenever it has the highest value ($1$) always wins, all other bidders do not win, the payment equals the winner's value, and there is no burn. If bidder $1$ has value $0$ and bidder $2$ has value $1$, then bidder $2$ always wins. If both have value $0$, then the lowest index bidder among the other bidders with a value of $1$ wins. By our characterization of Theorem~\ref{lem:single_item_char}, this is a $2$-SCP auction. Albeit, that characterization had some different assumptions (anonymous, continuous types). Nevertheless, it can be directly verified that the characterization still holds in this case. 

    Otherwise, if there is an input so that $C$ outputs $0$, then there are some bids by bidders $b_3, \ldots, b_{n+2}$ so that bidder $2$ wins when it bids a value of $0$ and bidder $1$ bids a value of $1$. By our construction of $C^1, C^2$, if bidder $1$ and $2$ switch their bids (bidder $1$ bids $0$, and bidder $2$ bids $1$), then bidder $1$ wins. Thus, there is a side-contract following Lemma~\ref{lem:highest_bidders}, where bidder $1$ and $2$ switch their bids, and allow the higher-value bidder win. 

\end{proof}

\section{A Separation of Active-Miner / Passive-Miner Models}
\label{sec:active-passive_miner}


We believe the passive miner model may be of independent interest in settings with more regulation over the auctioneer's actions, where bids are verified through some centralized registry and can not be dropped or faked easily.
In our main result, we have shown that the class of SCP mechanisms is the same whether the miner can drop and fake bids, or not. Since $2$-SCP equals SCP, this conclusion holds for $2$-SCP mechanisms as well. However, it does not hold for $1$-SCP mechanisms, as we now demonstrate.

\begin{definition}
    \textit{The Fully-Burned Second-Price Auction.} The auction confirms the single highest bidder, payment is the second highest bid, and the payment is fully burned. 
\end{definition}

Note that this is not the same as \cite{chung2023foundations}'s ``Burning Second-Price Auction", despite the similarity in the name. 

\begin{lemma}
\label{lem:active_passive_miner_separation}
    The Fully-Burned Second-Price Auction is $1$-SCP with passive miner, but not with active miner.
\end{lemma}
\begin{proof}

    Consider a collusion between the miner and a single bidder, that changes their bid from $b_i$ in $A$ to $b'_i$ in $B$. We do a case-analysis based on whether the bidder wins in $A$ and in $B$. 

    If the bidder loses in $B$, then their utility after the collusion is $0$. The miner utility is always $0$, so the joint utility after the collusion is $0$, and by individual rationality, this is not a beneficial collusion.

    If the bidder wins in both $A$ and $B$, then their utility does not change: Their confirmation is the same, and their payment is the same (as the second-price bid does not change). The miner utility is always $0$ and so overall the joint utility does not change, and this is not a beneficial collusion.

    If the bidder wins in $B$ but loses in $A$, then there exists a bid $b_j \geq b_i$ (since they lose in $A$), and so their utility after winning in $B$ is at most $b_i - b_j \leq 0$. 
    By individual rationality, this is not a beneficial collusion.

    The fully-burned second-price auction is UIC (as burns are irrelevant to UIC, and in terms of confirmation and payment it is equivalent to the single-item second-price auction, which is UIC). \cite{chung2023foundations} show there is no UIC+$1$-SCP mechanism for a finite number of items, where miners can be active. In particular, consider if the miner and the winning bidder collude so that the miner drops the second-highest bid, thus increasing the winner's utility (lower payment), but not affecting the zero miner's utility. 
\end{proof}

Notice that the result of Lemma~\ref{lem:single_item_1scp_2scp} is for the active-miner model. Thus, another way to articulate the result of Lemma~\ref{lem:active_passive_miner_separation} is that while in the active-miner model with a single item, $1$-SCP and $2$-SCP are equivalent, this is not true in the passive-miner model. 

\section{Theory and Evidence of Collusion in TFMs and Auctions}


By now, there is overwhelming evidence of miners engaging in strategic behavior. Miners strategically withhold transactions to increase fees in Bitcoin \cite{minersCollusionParlour}, and follow a pattern of colluding with each other in doing so. \cite{statisticalDetectionSelfishMining} finds mining cartels in Monacoin engage in what is known as a Selfish Mining attack \cite{selfishMining}. 
Miners extract value through MEV \cite{flashboys2, flashbots2022flashbots}, and miners use intricacies of the time-stamping protocol to increase revenue \cite{yaish2023uncle}. 
\cite{NFTSecurity} find evidence for shill-bidding used to increase miner revenue in NFT sales. They estimate the revenue generated in this way to be at least $13$ million USD. This directly relates to the SCP notion, as the form of shill bidding they consider falls under the definition of SCP (a colluding bidder issues bids with no intention of winning, to increase the payment to the miner in an English auction). They also find evidence for \textit{bid shielding}, which is more of a behavioral pattern: In an auction where bidders can retract their bids, a high bidder deters competition, and at the last minute pulls out so that a low bidder can win at a low price. The shielded bids amount to a total of about $940,000$ USD. While this is a form of user-user collusion, it is interesting that it is an order of magnitude less manifested, somewhat supporting the SCP's notion focus on miner-user collusion. However, \cite{candleAuctionsPolkadot}  observes user-user collusion in the Polkadot Parachains candle auctions, and show evidence that bidders engage in turn-taking, so as to jointly lower payments (`you avoid bidding against me, and I will avoid bidding against you').

There is extensive evidence document collusion in traditional auctions. First, there is the legal evidence: At some points (in particular, this was prevalent in the 1980s), more than $50\%$ of Antitrust cases in the US involved bid rigging in auctions \cite{porterZona1993}, and this remained a substantial issue thereafter \cite{krishna2002}. Second, there is the industry insiders / anecdotal evidence: \cite{grahamMarshall} mention an auctioneer who said ``that in 40 years of auctioneering, he
had yet to attend an auction at which a ring was not present''. Third, there is an extensive economic literature developing methods to identify collusion empirically through statistical tests, based on publicly known data. This approach relies on identifying ``tells'' that would arise in collusive bidding but not in competitive bidding. In user-user collusion, bidders may avoid undercutting one another \cite{musolf2025,kawaiNa}. Winning bids will be drawn from a distribution that corresponds to and makes sense as real valuations, but losing bids will not be constrained in this way \cite{porterZona1993}. Fourth, controlled lab experiments show that bidders that communicate and make transfers significantly reduce miner revenue, which is interpreted as facilitating user-user collusion \cite{agranovYariv}. 


Intuitively, there are several issues that prohibit a larger-size collusion. \cite{Stigler1964} raises the issue of the possible lack of \textit{trust} among the colluders. This may prohibit the collusion in at least two ways: 
(1) A colluder may implement a policy that violates the agreed-upon collusion. 
\cite{ferreira2021dynamic} explores this issue in the blockchain setting, by characterizing TFMs robust only to a collusion where it is incentive-compatible for bidders to follow the collusion, and not further deviate from it for self-gain. To counter that, \cite{smartCollusion} introduce the idea of \textit{smart collusion}, where colluders can impose penalties on defectors, in line with the models of many works in traditional economic theory. 

(2) A colluder may defect and report the existence of the collusion. \cite{harvardLawCartels} writes that the ``leniency program'', which allows members of cartels to self-report, "has become the most effective tool in the fight against cartels". 
Another issue in a larger-size collusion is that of \textit{bargaining}. Results in cooperative game theory \cite{compteJehiel} show that bargaining among a larger coalition results in higher loss of economic efficiency. 

Empirical evidence shows that a small-sized collusion is more prevalent than a larger-size one, across different works \cite{posner70,hayKelley74,FrassGreer77}. 

\end{document}

%% file: preamble.tex
\newif\ifcomments   
\commentsfalse
 \commentstrue

\renewcommand{\citeauthor}[1]{\citet{#1}}

\usepackage{csquotes}

\usepackage{amsmath}        
\usepackage{dsfont}         
\usepackage{mathtools}      
\usepackage{enumitem}       
\usepackage{amsthm}         
\usepackage{thmtools}
\usepackage{thm-restate}   
\usepackage{amsfonts}       
\usepackage[unicode]{hyperref} 
\usepackage{xcolor}         
\hypersetup{                
    colorlinks,
    linkcolor={red!50!black},
    citecolor={blue!50!black},
    urlcolor={blue!80!black}
}
\usepackage{etoolbox}\appto\UrlBreaks{\do\-}
\usepackage[capitalise]{cleveref}       

\usepackage{comment}
\ifcomments
\usepackage[colorinlistoftodos,prependcaption,textsize=tiny,textwidth=\marginparwidth]{todonotes}
\newcommand{\aviv}[1]{\ifcomments {\todo[color=blue!40,inline]{Aviv: #1}} \fi}
\newcommand{\yotam}[1]{\ifcomments {\todo[color=green!40,inline]{Yotam: #1}} \fi}
\else
\newcommand{\aviv}[1]{}
\newcommand{\yotam}[1]{}
\fi

\newtheorem{definition}{Definition}[section]

\newtheorem{claim}[definition]{Claim}

\newtheorem{lemma}[definition]{Lemma}
\newtheorem{corollary}[definition]{Corollary}

\newtheorem{remark}[definition]{Remark}
\newtheorem{example}[definition]{Example}

\crefname{definition}{Def.}{Defs.}
\crefname{theorem}{Theorem}{Theorems}
\crefname{claim}{Claim}{Claims}
\crefname{fact}{Fact}{Fact}
\crefname{lemma}{Lemma}{Lemmas}
\crefname{corollary}{Corollary}{Corollaries}
\crefname{example}{Example}{Examples}
\crefname{remark}{Remark}{Remarks}

\newcommand{\define}{\stackrel{\mathclap{\mbox{\text{\tiny def}}}}{=}}

\usepackage[symbols,acronym,nonumberlist,nogroupskip,section=subsection,numberedsection,stylemods={mcols,longbooktabs}]{glossaries-extra}

\glssetcategoryattribute{acronym}{nohyper}{true} 
\setabbreviationstyle[acronym]{long-short}  
\newacronym{DeFi}{DeFi}{decentralized finance}
\newacronym{PoW}{PoW}{Proof-of-Work}
\newacronym{PoS}{PoS}{Proof-of-Stake}
\newacronym{MEV}{MEV}{miner-extractable value}
\newacronym{block-DAG}{block-DAGs}{block directed-acyclic-graph}
\newacronym{PDF}{PDF}{probability density function}
\newacronym{CDF}{CDF}{cumulative density function}
\newacronym{AMM}{AMM}{automated market maker}
\newacronym{USD}{USD}{United States Dollar}
\newacronym{EIP}{EIP}{Ethereum improvement proposal}
\newacronym{mempool}{mempool}{memory pool}
\newacronym{TFM}{TFM}{transaction fee mechanism}
\newacronym{iid}{i.i.d.}{independent and identically distributed}
\newacronym{wrt}{w.r.t.}{with regard to}
\newacronym{wlog}{w.l.o.g.}{without loss of generality}
\newacronym{DSIC}{DSIC}{dominant strategy incentive-compatible}
\newacronym{BIC}{BIC}{Bayesian incentive-compatible}
\newacronym{MMIC}{MMIC}{myopic miner incentive-compatible}
\newacronym{OCA}{OCA}{off-chain agreement}
\newacronym{SCP}{SCP}{side-contract-proof}
\newacronym{PABGA}{PABGA}{pay-as-bid greedy auction}
\newacronym{SPA}{SPA}{second price auction}
\newacronym{UPGA}{UPGA}{uniform-price greedy auction}
\newacronym{BNE}{BNE}{Bayesian-Nash equilibrium}
\newacronym{EPIR}{EPIR}{ex-post individually rational}
\newacronym{EPBB}{EPBB}{ex-post burn balanced}
\newacronym{GTA}{GTA}{good toy auction}
\newacronym{GTFBA}{GTFBA}{good toy finite-blocksize auction}
\newacronym{QoS}{QoS}{quality of service}
\newacronym{CTPA}{CTPA}{constant total probability of allocation}

\glsxtrnewsymbol[description={
    An auction.
}]{auction}{
    \ensuremath{A}
}
\newcommand{\auction}{{\gls[hyper=false]{auction}}}

\glsxtrnewsymbol[description={
    A transaction.
}]{tx}{
    \ensuremath{tx}
}

\glsxtrnewsymbol[description={
    Payment rule.
}]{pay}{
    \ensuremath{p}
}
\newcommand{\pay}{{\gls[hyper=false]{pay}}}

\glsxtrnewsymbol[description={
    Burning rule.
}]{burn}{
    \ensuremath{\beta}
}
\newcommand{\burn}{{\gls[hyper=false]{burn}}}

\glsxtrnewsymbol[description={
    Update rule.
}]{update}{
    \ensuremath{\upsilon}
}

\glsxtrnewsymbol[description={
    Reserve price of an auction.
}]{reserve}{
    \ensuremath{r}
}
\newcommand{\reserve}{{\gls[hyper=false]{reserve}}}

\glsxtrnewsymbol[description={
    Allocation function.
}]{allocation}{
    \ensuremath{a}
}
\newcommand{\alloc}{{\gls[hyper=false]{allocation}}}

\glsxtrnewsymbol[description={
    Transaction fee of some transaction, in tokens.
}]{fee}{
    \ensuremath{b}
}
\newcommand{\fee}{{\gls[hyper=false]{fee}}}

\glsxtrnewsymbol[description={
    Predefined maximal block-size, in bytes.
}]{blocksize}{
    \ensuremath{\mathcal{B}}
}
\newcommand{\blocksize}{{\gls[hyper=false]{blocksize}}}

\glsxtrnewsymbol[description={
    Miner revenue.
}]{revenue}{
    \ensuremath{u}
}

\glsxtrnewsymbol[description={
    Number of all bids with $b_i \geq \reserve$.
}]{bidsMoreReserve}{
    \ensuremath{S_{\geq \reserve}}
}

\NewDocumentCommand{\expect}{ e{_} s o >{\SplitArgument{1}{|}}m }{%
  \operatorname{E}
  \IfValueT{#1}{{\!}_{#1}}
  \IfBooleanTF{#2}{
    \expectarg*{\expectvar#4}%
  }{
    \IfNoValueTF{#3}{
      \expectarg{\expectvar#4}%
    }{
      \expectarg[#3]{\expectvar#4}%
    }%
  }%
}
\NewDocumentCommand{\expectvar}{mm}{%
  #1\IfValueT{#2}{\nonscript\;\delimsize\vert\nonscript\;#2}%
}
\DeclarePairedDelimiterX{\expectarg}[1]{[}{]}{#1}